\documentclass[11pt]{article}

\usepackage{fullpage}

\usepackage{hyperref}
\usepackage{url}
\usepackage{amsfonts} 
\usepackage{xspace}
\usepackage{bm}
\usepackage{algorithm}
\usepackage[noend]{algpseudocode}
\usepackage{amsmath}
\usepackage{amssymb} 
\usepackage{amsthm}
\usepackage{cleveref}
\usepackage{mathtools}
\usepackage{tikz}
\usepackage{wrapfig}
\usepackage{thm-restate}
\usepackage{nicefrac}
\usepackage[style=alphabetic,natbib=true,maxcitenames=2, maxbibnames=10,backend=bibtex]{biblatex}

\hypersetup{
	colorlinks = true,
	linkcolor = black,
	citecolor = black,
	linktocpage = true,
	urlcolor = black
}

\newtheorem{lemma}{Lemma}
\newtheorem{theorem}{Theorem}
\newtheorem{definition}{Definition}
\newtheorem{claim}{Claim}

\newcommand{\Reals}{\mathbb{R}}
\newcommand{\GFT}{\textnormal{GFT}}
\newcommand{\Tau}{\mathcal{T}}
\def\ie{\textit{i.e.}\@\xspace}
\def\eg{\textit{e.g.}\@\xspace}
\def\blindguess{\textsc{Guess}\xspace}
\def\reducearea{\textsc{Reduce}\xspace}
\def\geometricreducearea{\textsc{GeometricReduce}\xspace}
\def\geometricguess{\textsc{GeometricGuess}\xspace}
\def\leaf{\textsc{Leaf}\xspace}
\newcommand{\area}{\textup{\textbf{Region}}}
\newcommand{\child}{\mathcal{C}}
\newcommand{\pred}{\textup{\textbf{pred}}}

\usepackage{color}

\algrenewcommand{\algorithmiccomment}[1]{\textcolor{gray}{\hskip3em$\#$ #1}}

\title{Nonparametric Contextual Online Bilateral Trade}

 \author{
 Emanuele Coccia\thanks{Bocconi University, \texttt{emanuele.coccia@studbocconi.it}}
 \and 
 Martino Bernasconi\thanks{Bocconi University, \texttt{martino.bernasconi@unibocconi.it}. Corresponding author.}
 \and 
 Andrea Celli\thanks{Bocconi University, \texttt{andrea.celli2@unibocconi.it}}
 }

\begin{document}

\maketitle

\begin{abstract}
We study the problem of contextual online bilateral trade. At each round, the learner faces a seller-buyer pair and must propose a trade price without observing their private valuations for the item being sold. The goal of the learner is to post prices to facilitate trades between the two parties. Before posting a price, the learner observes a $d$-dimensional context vector that influences the agent's valuations. Prior work in the contextual setting has focused on linear models. In this we tackle a general nonparametric setting in which the buyer’s and seller’s valuations behave according to arbitrary Lipschitz functions of the context. We design an algorithm that leverages contextual information through a hierarchical tree construction and guarantees regret $\widetilde{O}(T^{{(d-1)}/d})$. 
Remarkably, our algorithm operates under two stringent features of the setting: (1) one-bit feedback, where the learner only observes whether a trade occurred or not, and (2) strong budget balance, where the learner cannot subsidize or profit from the market participants.
We further provide a matching lower bound in the full-feedback setting, demonstrating the tightness of our regret bound.
\end{abstract}

\section{Introduction}

Bilateral trade is a fundamental economic model which describes the interaction between a seller and a buyer, each with private valuations for a good, who seek to engage in trade with the goal of maximizing their individual utilities \citep{Vickrey61,MyersonS83}. This scenario arises in many applications such as ridesharing platforms or energy and financial exchanges.

Recently, \citet{cesa2024bilateral} introduced the online bilateral trade problem, in which at each round $t$, a new seller and buyer arrive with private valuations $s_t$ and $b_t$, respectively. The valuation of the seller $s_t$ is the lowest price they are willing to sell the item. Analogously, the valuation of the buyer, $b_t$, is the highest price the buyer is willing to pay for the item. 
At each round $t$, the learner must propose a price to both the buyer and the seller without observing their private valuations. 
After posting price $p_t$, a trade happens if $s_t\le p_t\le b_t$. The \emph{gain from trade} at time $t$ is
\begin{equation}\label{eq:gft}
\GFT(p_t|s_t,b_t)\coloneqq \mathbb{I}(s_t\le p_t\le b_t)(b_t-s_t),
\end{equation}
and represents the total increase in social welfare generated by the trade. The goal of the learner is to maximize the cumulative gain from trade over time, which is equivalent to minimizing the regret with respect to the best policy in hindsight. 

In many practical scenarios, the learner has access to side information at each round $t$ about the seller, buyer, or the item being traded. For instance, this is common in online marketplaces, where items are highly differentiated and the learner can observe certain features of the item prior to setting a price. Such contextual information can be leveraged to estimate the relevance of different features based on historical data.
\citet{gaucherfeature} introduced the \emph{feature-based} online bilateral trade problem, in which the learner observes a $d$-dimensional feature vector $x_t$ before posting a price at round $t$. They study a linear model where private valuations are linear functions of the form $x_t^\top \theta$ for some unknown parameter vector $\theta$.
In this work, we study whether it is possible to move beyond the linear valuation model. A similar question was examined in the context of pricing in \citet{mao2018contextual}.

We provide a comprehensive characterization of the problem when valuations are \emph{arbitrary} $L$-Lipschitz functions of the context vectors. %
A key feature of our analysis is that it operates under the \emph{one-bit feedback} model: at each round, the learner observes only the signal $\mathbb{I}(s_t \le p_t \le b_t)$, indicating whether a trade occurred or not. This is the most information-constrained setting in online bilateral trade, and existing contextual no-regret algorithms can handle this feedback model only by relaxing the budget balance constraint~\citep{gaucherfeature}. 

\subsection{Technical Challenges and Techniques}

Compared to classic online learning problems, the bilateral-trade setting presents additional technical difficulties, making traditional approaches for nonparametric contextual online learning largely inapplicable.
The most notable is that the reward function, which is the gain from trade (\Cref{eq:gft}), is non-continuous both as a function of $p$ or as a function of the valuations $s$ and $b$.
The second challenge, and possibly a more fundamental one, is the lack of observability of the reward we are trying to maximize; indeed, valuations $s$ and $b$ remain hidden, whether the trade happens or not, and so does the $\GFT$. The only feedback received by the algorithm is the indicator of $\mathbb{I}(s\le p\le b)$, which only signals whether the trade was accepted or not.

Although our problem has features that separate it from classic online learning, we do not need to consider a generic Lipschitz loss function, and we will heavily exploit the structure of the reward $p\mapsto \GFT(p|s,b)$ and the Lipschitzness of the functions generating $s,b$ from the contexts. The specific shape of the $\GFT$ function (and the absence of noise) is what ultimately allows us break the lower bound of $\Omega(T^{(d+2)/(d+3)})$ \citep[from][with $p=1$]{lu2009showing, slivkins2014contextual} that we would obtain if our reward function were a generic Lipschitz function (see \Cref{sec:RW} for more details).

Similar to the \emph{adaptive zooming} technique employed by \citet{slivkins2014contextual} or the \emph{chaining} approach of \citet{cesa2017algorithmic}, we maintain a discretized representation of the context space, which is geometrically refined in regions in which we observe a lot of contexts. Formally, we use a tree with branching factor $2^d$, in which each layer $\ell$ maintains an $O(2^{-\ell})$ discretization of the function mapping contexts to optimal prices. Moreover, at level $\ell$, we ``mark'' a node if we find a price that was accepted when observing a context in the region described by that node. By Lipschitzness, we know that there is a region of size $O(L2^{-\ell})$ that contains a price for all nodes of the sub-tree.

The main challenge then is to effectively utilize this information, as the discontinuity of the reward functions prevents the use of standard techniques. The key trick here is to heavily rely on randomization when posting prices.
By playing uniformly at randomly on an $\epsilon$ discretization of the identified set of good prices (that has length $O(2^{-\ell})$) we have a probability of at least $\Omega(2^\ell\epsilon)$ of winning any trade with $\GFT$ at least $\epsilon$, and thus we will waste only $O({2^{-\ell}}/{\epsilon})$ turns on expectation before finding a price that is accepted, marking the node with that price and continuing to the children node. On the other hand, losing all the trades with $\GFT$ smaller than $\epsilon$ only yields a regret of order $O(\epsilon T)$. By tuning the parameter $\epsilon$, and summing over all node of the tree this leads to an overall regret of $\tilde O(T^{{(d-2)}/{d}})$, which, although being sublinear, is suboptimal.

In order to obtain an optimal rate, the second trick is associated with reducing the regret suffered when guessing the right price and exploits the geometric structure of the problem. Before starting to guess a correct price in the associated interval, we perform a sort of ``hide-and-seek game'' against the adversary generating the contexts. We probe iteratively two prices. If we win the trade, we are happy since we have zero regret. Otherwise, if we miss a trade for both prices, by geometric reasoning, we can limit the space in which the context can place the valuation, and we can show it has to be very close (order of $O(2^{-\ell})$) to the diagonal. This limits the possible $\GFT$ (and thus the regret) by $O(2^{-\ell})$ in all subsequent nodes. This modifies the maximum regret experienced in the guessing phase to be $O({2^{-2\ell}}/{\epsilon})$ at nodes at layer $\ell$, and by appropriately tuning $\epsilon$ we can obtain an optimal regret of $O(T^{{(d-1)}/{d}})$.

Crucially, this algorithm relies on explicit knowledge of the Lipschitz constant $L$, which may be unknown in a practical scenario. Building on the main structure of the vanilla algorithm, in \Cref{sec:noL}, we design an algorithm that works by testing geometrically increasing scales of the Lipschitz constant $L$ and only slightly worsens the guarantees, without requiring knowing the Lipschitz constant.

For the matching (up to logarithmic terms) lower bound of $\Omega(T^{(d-1)/d})$, we use a combination of Yao's weak minimax principle \citep{yao1977probabilistic}, which relates the performance of randomized algorithms in the worst instance to the performance of deterministic algorithms against a distribution over instances, with a simple version of the McShane's extension theorem \citep{mcshane1934extension} that shows the existence of Lipschitz extensions on the continuum of Lipschitz functions on discrete sets.

\subsection{Further Related Works}\label{sec:RW}

\paragraph{Bilateral trade} In the offline setting, \citet{MyersonS83} demonstrated that no mechanism can simultaneously achieve full efficiency, incentive compatibility, individual rationality, and budget balance in general. This impossibility result motivated a long line of research aimed at designing mechanisms that achieve approximate efficiency in Bayesian settings \citep{Mcafee08,BlumrosenM16,brustle2017approximating,DengMSW21,kang22fixed}.

\paragraph{Online bilateral trade} In the online setting, \citet{cesa2024bilateral} study a setting where buyer and seller valuations are i.i.d.~from an unknown distribution, and the learner must enforce strong budget balance (\ie, at each round they post the same price to both the seller and the buyer). They provide no-regret algorithms for the full-feedback setting, and under partial feedback when the valuations are drawn independently for the seller and the buyer, and the underlying distribution is smooth. 
In the case in which valuations are generated by an adversary, \citet{AzarFF22} provides an algorithm achieving a tight sublinear $2$-regret under weak budget balance (\ie, $p_t \le q_t$ at each round $t$). Moreover, \citet{bernasconi2023no} shows that sublinear regret is still achievable if the learner is allowed to enforce global budget balance, meaning the constraint must hold over the entire horizon rather than at each round (see also \citet{chen2025tight, lunghi2025better}).
The most closely related work in online bilateral trade is by \citet{gaucherfeature}, who consider a contextual variant of the problem where valuations at time $t$ take the form $x_t^\top \theta$, possibly with additive noise. They propose an algorithm achieving $\widetilde{O}(T^{2/3})$ regret under two-bit feedback and strong budget balance, and show that sublinear regret is also attainable under one-bit feedback if global budget balance is allowed.
This linear contextual valuation model builds on a substantial body of work in contextual pricing with linear valuations \citep{Kleinberg2003,amin2014repeated,Cohen_feature_DP,lobel2018multidimensional,liu2021optimal}.

\paragraph{Other related models} A related but structurally different model is that of online brokerage \citep{bolic2024online}, which differs from standard bilateral trade in that traders may act as buyers or sellers depending on market conditions, and their valuations are drawn i.i.d. from a shared unknown distribution. \citet{bachoc2024contextual} consider a contextual brokerage model, in which valuations are modeled as zero-mean perturbations of a market price that is linear in the agent's feature vector.

\paragraph{Online nonparametric contextual online learning}
The study of online learning with nonparametric contexts was initiated by \citet{hazan2007online}, in which they proved an upper bound of $\widetilde O(T^{(d+1)/(d+2)})$ and context dimension $d$, while providing a $\Omega(T^{(d-1)/d})$ lower bound. The gap was closed in \citet{rakhlin2015online} for full information by providing an upper bound of $\widetilde O(T^{(d-1)/d})$ and later refined in \citet{cesa2017algorithmic} by providing an explicit algorithm. Both of these works also consider an arbitrary Lipschitz reward function. The model was extended to arbitrary metric action spaces of dimension $p$ in \citet{lu2009showing} and \citet{slivkins2014contextual}, that show an almost matching rate of $\Theta(T^{(d+p+1)/(d+p+2)})$.
In economic settings, the problem of nonparametric contextual information has been studied in \citet{cesa2017algorithmic}, which examines second-price auctions, and \citet{mao2018contextual, chen2021nonparametric, tullii2024improved} for pricing.

\paragraph{Independent and concurrent work} Independently, \citet{Cosson2026better} considers the contextual bilateral trade model with linear contexts of \citet{gaucherfeature}, focusing on the noiseless setting. They study gain from trade and profit maximization (unlike \citet{gaucherfeature} which only study gain from trade) in various feedback settings and budget balance constraints. 
The linear noiseless contextual setting allows for substantially better rates with respect to what is achievable in our work. In particular, the rates they obtain are at most logarithmic in the time horizon and exponential in the dimension, while ours are polynomial with exponent depending on the dimension. %

\section{Preliminaries}

At the beginning of each round $t$, the learner observes a public context $x_t\in[0,1]^d$ generated by an oblivious adversary. At the same time, a new pair of buyer and seller arrives, characterized by private valuations $(b_t,s_t)\in[0,1]^2$. The valuations are generated from the context by two unknown functions $s_t=f_s(x_t)$ and $b_t=f_b(x_t)$. We assume, as it is standard in the nonparametric contextual bandits literature (see, \eg, \citep{slivkins2014contextual,cesa2017algorithmic}), that $f_s,f_b:[0,1]^d\to[0,1]$ are $L$-Lipschitz functions, \ie, $|f_s(x)-f_s(x')|\le L\|x-x'\|_\infty$ for all $x,x'\in[0,1]^d$ (and similarly for $f_b$). We denote the class of all $L$-Lipschitz functions as $\mathcal{F}_{L}$.

\paragraph{Objective} After observing the contextual information $x_t$, the learning picks one price $p_t\in[0,1]$, and the trade is accepted if both the seller and the buyer accept it, \ie, $s_t\le p_t\le b_t$.
When the trade is accepted, it generates a social welfare of $b_t-s_t$ (that is, the increase in total happiness between the two market participants); in particular, the \emph{gain from trade} is defined as in \Cref{eq:gft}.
The objective is to minimize the regret, which given a sequence of contexts $\bm{x}=\{x_1,\ldots,x_T\}\in([0,1]^d)^T$ and two $L$-Lipschitz functions $f_s,f_b$, is defined as
\[
R_T(f_s,f_b,\bm{x})=\sum_{t=1}^T[f_b(x_t)-f_s(x_t)]^+-\mathbb{E}[\GFT(p_t|f_s(x_t),f_b(x_t))],
\]
where the expectation is taken over the randomization of the algorithm.
The regret of an algorithm is then defined as $R_T=\sup_{f_s,f_b\in\mathcal{F}_L,\bm{x}\in([0,1]^d)^T}R_T(f_s,f_b,\bm{x})$, which is the worst-case regret attainable across all functions $f_s,f_b\in \mathcal{F}_L$ and all context sequences $\bm{x}$.
Since we are interested in regret guarantees that are asymptotic in $T$, we will work under the assumption that $T\gg 2^d$. In particular, if $T\le 2^d$, we can only guarantee an upper bound on the regret of $2^d$.

\paragraph{Feedback} Our algorithm operates under partial feedback, specifically under \emph{one-bit feedback}. In this feedback model, the learner observes only a single bit of information indicating whether a trade occurred, \ie, $\mathbb{I}(s_t \le p_t \le b_t)$. In contrast, the lower bound is established under \emph{full feedback}, where the learner observes the valuation pair $(s_t, b_t)$ after posting the price $p_t$, irrespective of whether the trade was executed. An intermediate setting is the \emph{two-bit feedback} model, in which the learner observes both $\mathbb{I}(s_t \le p_t)$ and $\mathbb{I}(p_t \le b_t)$.

\paragraph{Notation}
We denote by $H_d$ the $d$-dimensional Boolean hypercube $\{0,1\}^d$, and by $\bm{0}_d$ a $d$-dimensional vector of zeros.
We denote the projection of $x\in\mathbb{R}^n$ onto $\mathcal{C}\subset \mathbb{R}^n$ as $\Pi_{\mathcal{C}}(x)$.
Moreover, for $a \le b$, let $[a,b]_\epsilon$ denote the uniform grid over the interval $[a,b]$ with step size $\epsilon$, where any grid point outside the interval $[0,1]$ is projected onto the nearest endpoint, \ie, capped to $0$ if below $0$ and to $1$ if above $1$. Formally, 
\(
[a,b]_\epsilon\coloneqq\Pi_{[0,1]}\left(\{a,a+\epsilon, \ldots, a+k\epsilon, b\}\right)
\), where $k=\lfloor\frac{b-a}{\epsilon}\rfloor$.

\section{Regret guarantees for Known Lipschitz Constant}\label{sec:UB}

In this section, we present an algorithm that guarantees regret of order $\widetilde O(LT^{(d-1)/d})$ in the case in which the Lipschitz parameter $L$ is known. In \Cref{sec:noL} we extend this approach to include the case in which $L$ is not known.

\subsection{Hierarchical Tree Structure}

We start by defining a hierarchical partitioning of the $d$-dimensional hypercube. We define a tree $\Tau$ in which each node represents a specific subset of the hypercube. A node $N_{\ell, z}$ in the tree is defined by a level $\ell$ and a reference point $z$. The subset of the hypercube defined by a node $N_{\ell, z}$ is 
\[
\area(N_{\ell, z})\coloneqq[z_1,z_1+2^{-\ell}]\times \ldots\times [z_d,z_d+2^{-\ell}].
\]

The root of the tree corresponds to $\ell=0$ and the reference point $z=\bm{0}_d$. Each node $N_{\ell, z}$ has $2^d$ children $N_{\ell+1,z_i}$, where $z_i=z+{2^{-(\ell+1)}}h_i$ and $h_i$ ranges over all $2^d$ binary vectors.
We denote by $\child(\ell,z)\subset [0,1]^d$ the set of all the reference points corresponding to the children of $N_{\ell, z}$. 
Moreover, for any $z'\in \child(\ell,z)$, we use $\pred(\ell+1,z')$ to denote the reference point of the parent node of $N_{\ell+1,z'}$ (\ie, point $z$). Under this notation, the parent node of $N_{\ell,z}$ is $N_{\ell-1,\pred(\ell,z)}$.

We denote the set of reference points in nodes at level $\ell$ with $Z_\ell$, which is defined recursively as $Z_0=\bm{0}_d$ and $Z_\ell = \cup_{z\in Z_{\ell-1}}\child(\ell,z)$.
Clearly, $\area(N_{\ell,z}) = \bigcup_{z' \in \child(\ell,z)} \area(N_{\ell+1,z'})$. Then, each level $\ell$ of the tree induces a partition of the $d$-dimensional hypercube, since by construction, $\bigcup_{z \in Z_\ell} \area(N_{\ell,z}) = [0,1]^d$. We limit the height of the tree to $H=\lfloor \log_{2^d}(T)\rfloor=O(\log(T)d^{-1})$.

\subsection{Algorithm}\label{sec:algo}

\begin{algorithm}[t]
\caption{}\label{alg:EPBT}
\begin{algorithmic}[1]
\Require Horizon $T$, Lipschitzness $L$, context dimension $d$.
\For{ $t\in [T]$}
\State Observe context $x_t$
\State Initialize $\ell\gets 0,z\gets \bm{0}_d$ %
\While{$p_{\ell,z}\neq \varnothing$ and $N_{\ell,z}$ not leaf}\hspace{-0.5cm}\hfill\Comment{Traverse tree until  active or leaf node is reached}\label{line:traverse}
\State Find $z'$ such that $x_t\in \area(N_{\ell+1,z'})$, where $z'\in\child(\ell,z)$
\State $\ell\gets\ell+1$, $z\gets z'$
\EndWhile
\If{$\reducearea(N_{\ell,z})$ has not terminated}
\State Execute one iteration of the main loop of $\reducearea(N_{\ell,z})$\label{line:reduce}
\ElsIf{$\blindguess(N_{\ell,z})$ has not terminated}
\State Execute one iteration of the main loop of $\blindguess(N_{\ell,z})$\label{line:blind}
\If{$\blindguess(N_{\ell,z})$ terminates returning price $p$}
\State Set node $N_{\ell,z}$ as marked by setting $p_{\ell,z}\gets p$\label{line:mark}
\EndIf
\EndIf
\State \algorithmicelse~Play $p_{\ell,z}$ \hfill\Comment{If node is a marked leaf play corresponding price}\label{line:markedleaf}
\EndFor
\end{algorithmic}
\end{algorithm}

The algorithm works as outlined in \Cref{alg:EPBT}. In particular, each node $N_{\ell,z}$ of the tree may be ``marked with'' a price $p_{\ell,z} \in [0,1]$ (Line \ref{line:mark}).

\begin{definition}\label{def:marked}
A node $N_{\ell,z}$ is said to be \emph{marked} if it has been assigned a \emph{marking price} $p$ such that, for some context $x \in \area(N_{\ell,z})$, the trade is accepted at $p$, that is, $\mathbb{I}(f_s(x) \le p \le f_b(x)) = 1$.
\end{definition}

A node $N_{\ell,z}$ can be marked only if every node along the path from the root to $N_{\ell,z}$ has already been marked.
The set of marked nodes, therefore, forms a rooted subtree. The free (\ie, not marked) leaves of this subtree are referred to as \emph{active nodes}. 
Intuitively, when we reach an active node $N_{\ell,z}$, our goal is to find a price $p$ that results in a successful trade, that is, one satisfying $\mathbb{I}(f_s(x_t) \le p \le f_b(x_t))$. To this end, we leverage the information already encoded in the tree: every node $N_{\ell',z'}$ along the path from the root to $N_{\ell,z}$ is marked with a price $p_{\ell',z'}$ at which a trade occurred. These previously observed prices allow us to restrict the interval of candidate prices and form an estimate for a price $p$ that will be accepted by the valuations $f_s(x_t),f_b(x_t)$. 

\begin{algorithm}[t]
\caption{$\reducearea(N_{\ell,z})$}\label{alg:reducearea}
\begin{algorithmic}[1]
\State Set $\bar p\gets p_{\ell-1,\pred(\ell,z)}$ (\ie, the price associated with the parent node) \hspace{-0.5cm}\hfill\Comment{Initialization}
\State Set $p_L\gets\Pi_{[0,1]}\left(\bar p-L2^{-(\ell-1)}\right)$, $p_U\gets\Pi_{[0,1]}\left(\bar p+L2^{-(\ell-1)}\right)$
\State Initialize $p\gets p_L$ 
\Repeat \hfill\Comment{Main Loop}
\State Post price $p$ and observe $\mathbb{I}(s_t\le p \le b_t)$
\If{$\mathbb{I}(s_t\le p \le b_t)=0$}
\State \algorithmicif~$p = p_L$ \algorithmicthen~set $p \gets p_U$; \algorithmicelse~terminate $p\gets\varnothing$\label{line:pLpU}
\EndIf
\Until{$p=\varnothing$}
\end{algorithmic}
\end{algorithm}

\begin{algorithm}[t]
\caption{$\blindguess(N_{\ell,z})$}\label{alg:guess}
\begin{algorithmic}[1]
\State Let $\bar p=p_{\ell-1,\pred(\ell,z)}$ be the price associated to the parent node $N_{\ell-1,\pred(\ell,z)}$ \hspace{-1.1cm}\hfill\Comment{Initialization}
\State Set $\mathcal{P}_{\ell,z}=[\bar p-L2^{-(\ell-1)}, \bar p+L2^{-(\ell-1)}]_{\epsilon}$ with $\epsilon=LT^{-\nicefrac{1}{d}}$
\Repeat \hfill\Comment{Main Loop}
\State Post price $p\sim\text{Unif}(\mathcal{P}_{\ell,z})$ and observe $\mathbb{I}(s_t\le p \le b_t)$
\Until{$\mathbb{I}(s_t\le p \le b_t)=1$}
\State Return $p$
\end{algorithmic}
\end{algorithm}

\paragraph{Node-wise regret decomposition} In the analysis of the algorithm, we decompose the total regret into the contributions to the regret of each node. Let $A \subseteq [T]$ be a subset of rounds. Then, we define
\[
\mathcal{R}(A)\coloneqq \sum_{t\in A} \left([f_b(x_t)-f_s(x_t)]^+-\GFT_t(p_t|f_s(x_t),f_b(x_t))\right).
\]
For any node $N_{\ell,z} \in \Tau$, we define $T_{\ell,z} \subseteq [T]$ as the set of rounds in which the traversal of the tree stops on $N_{\ell,z}$ (\ie, the number of rounds in which $N_{\ell,z}$ is the node selected at the end of the loop of Line \ref{line:traverse}). Then, for any node $N_{\ell,z}$, we define the regret ``accumulated on $N_{\ell,z}$'' as the regret accumulated in rounds $T_{\ell,z}$: $R_{\ell, z}=\mathbb{E}\left[\mathcal{R}(T_{\ell,z})\right]$.
Then, the overall regret can be rewritten by separating contributions incurred at each level of the tree and at the nodes within that level:
\begin{equation}\label{eq:nodeDecomposition}\textstyle
R_T = \sum_{\ell=0}^{H} \sum_{z\in Z_\ell} R_{\ell, z}.
\end{equation}

Any node of the tree that is never visited by \Cref{alg:EPBT} has $T_{\ell,z} = \emptyset$ and thus does not contribute to the overall regret.

For each observed context $x_t$, we traverse the tree $\Tau$ by repeatedly selecting the child node whose associated region contains $x_t$, until reaching an active node or a leaf node (see Line \ref{line:traverse} of \Cref{alg:EPBT}). 
Then, the algorithm proceeds in one of two ways: (i) if the node is not marked, we invoke one of the two routines \reducearea and \blindguess described below (Lines \ref{line:reduce} and \ref{line:blind}); (ii) if the node is a marked leaf, we post the price previously assigned to it (Line \ref{line:markedleaf}).

\subsection{\reducearea routine}

Consider an active node $N_{\ell,z}$ and its parent $N_{\ell-1,\pred(\ell,z)}$, and focus on the rounds $t\in[T]$ such that $x_t\in \area(N_{\ell,z})$. Whenever \Cref{alg:EPBT} selects such node after the tree traversal, the algorithm executes one iteration of the \reducearea routine on $N_{\ell,z}$ until its termination. The number of rounds in which \Cref{alg:EPBT} calls $\reducearea(N_{\ell,z})$ is denoted by $T_{\ell,z}^\reducearea$.

The \reducearea routine is described in \Cref{alg:reducearea}. Starting from the price $\bar p= p_{\ell-1,\pred(\ell,z)}$ associated to the parent node of $N_{\ell,z}$, the algorithm computes two prices centered around $\bar p$, each with distance $L2^{-(\ell-1)}$. We denote such prices by $p_L=\Pi_{[0,1]}(\bar p-L2^{-(\ell-1)})$ and $p_U=\Pi_{[0,1]}(\bar p+L2^{-(\ell-1)})$, where the projection onto $[0,1]$ ensures that prices are well-defined. Then, \reducearea repeatedly posts the price $p_L$ until a trade is rejected (\ie, $\mathbb{I}(s_t\le p_L \le b_t)=0$). It then switches to posting $p_U$ until that price is also rejected by one of the agents. Once both rejections occur, the routine for node $N_{\ell,z}$ terminates. From that point on, whenever the traversal ends at $N_{\ell,z}$, \Cref{alg:EPBT} invokes the \blindguess subroutine (see \Cref{sec:guess}).

This procedure may appear counterintuitive, as it intentionally seeks rejections of some specific prices. However, the primary objective of \reducearea is to control the regret that may accumulate once the routine concludes. Specifically, upon termination, both prices $p_L$ and $p_U$ must have been rejected once. Then, due to the Lipschitz continuity of the valuation functions, this restricts the range of future contexts for which the node $N_{\ell,z}$ can be selected, effectively ensuring that $f_b(x) - f_s(x) \le O(L 2^{-\ell})$ for any context $x$ in the region corresponding to $N_{\ell,z}$. Formally

\begin{restatable}{lemma}{lemmaafterRA}\label{lem:afterRA}
    Consider a node $N_{\ell,z}$ for which the \reducearea procedure terminated. Then, for any $x\in \area(N_{\ell,z})$, it holds that
    $f_b(x)-f_s(x)\le 6L2^{-\ell}$.
\end{restatable}

Although simple, this result is rather surprising, as it reveals a fundamental property of the problem. By viewing the context-generation process as adversarial, the previous result can be interpreted as follows: if the adversary ``hides'' the valuations both when $p_L$ and $p_U$ are played, then the Lipschitz constraint forces the adversary to place them close to the diagonal. In particular, the \reducearea procedure may not end for some nodes (indeed, for most nodes it is not expected to do so). However, we show that if the adversary forces termination at a specific node, then across all similar contexts (\ie, nodes in the subtree rooted at that node) the adversary's gain remains limited.
This situation is illustrated in \Cref{fig:boat3}, which also serves as intuition about the proof.

Now, we show that the regret incurred while running this routine is suitably upper-bounded. Formally, let 
$
R_{\ell,z}^{\reducearea}\coloneqq \mathbb{E}[\mathcal{R}(T_{\ell,z}^\reducearea)]
$ 
be the regret accumulated by the \reducearea routine executed on node $N_{\ell,z}$. We can upper bound this term as follows. 
\begin{restatable}[Regret due to \reducearea]{lemma}{lemmaregretRA}\label{lem:regretRA}
For any node $N_{\ell,z}$ with $\ell\ge 1$ it holds
    $
    R_{\ell,z}^\reducearea\le 24 L2^{-\ell}.
    $
\end{restatable}

\subsection{\blindguess routine}\label{sec:guess}
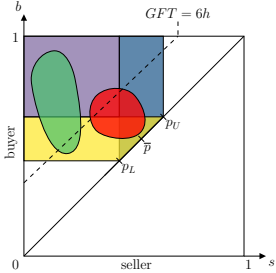
\begin{wrapfigure}[21]{l}{0.36\textwidth}
\vspace{-0.6cm}
  \centering\scalebox{0.37}{\tikzset{every picture/.style={line width=0.75pt}} %
\tikzstyle{every node}=[font=\Large]

\begin{tikzpicture}[x=0.75pt,y=0.75pt,yscale=-1,xscale=1]

\draw  [fill={rgb, 255:red, 49; green, 103; blue, 149 }  ,fill opacity=0.78 ] (380,200) -- (320,260) -- (320,90) -- (380,90) -- (380,171.24) -- cycle ;
\draw  [fill={rgb, 255:red, 254; green, 231; blue, 38 }  ,fill opacity=0.7 ] (380,200) -- (320,260) -- (288.83,260) -- (190,260) -- (190,200) -- cycle ;
\draw  [fill={rgb, 255:red, 71; green, 39; blue, 120 }  ,fill opacity=0.5 ] (190,90) -- (320,90) -- (320,200) -- (190,200) -- cycle ;
\draw   (190,90) -- (490,90) -- (490,390) -- (190,390) -- cycle ;
\draw    (320,260) -- (190,390) ;
\draw [shift={(320,260)}, rotate = 315] [color={rgb, 255:red, 0; green, 0; blue, 0 }  ][line width=0.75]    (0,5.59) -- (0,-5.59)   ;
\draw    (350,230) -- (300,280) ;
\draw [shift={(350,230)}, rotate = 315] [color={rgb, 255:red, 0; green, 0; blue, 0 }  ][line width=0.75]    (0,5.59) -- (0,-5.59)   ;
\draw    (380,200) -- (350,230) ;
\draw [shift={(380,200)}, rotate = 315] [color={rgb, 255:red, 0; green, 0; blue, 0 }  ][line width=0.75]    (0,5.59) -- (0,-5.59)   ;
\draw    (490,90) -- (380,200) ;
\draw  [dash pattern={on 4.5pt off 4.5pt}]  (400,70) -- (400,90) ;
\draw  [fill={rgb, 255:red, 255; green, 2; blue, 0 }  ,fill opacity=0.76 ] (290,170) .. controls (304.82,157.28) and (329.82,159.28) .. (340,170) .. controls (350.18,180.72) and (362.18,204.72) .. (350,220) .. controls (337.82,235.28) and (302.18,229.72) .. (290,220) .. controls (277.82,210.28) and (275.18,182.72) .. (290,170) -- cycle ;

\draw  [fill={rgb, 255:red, 73; green, 193; blue, 109 }  ,fill opacity=0.71 ] (230,110) .. controls (246.82,108.28) and (256.82,132.28) .. (260,150) .. controls (263.18,167.72) and (279.82,245.28) .. (250,250) .. controls (220.18,254.72) and (200.82,169.28) .. (200,150) .. controls (199.18,130.72) and (213.18,111.72) .. (230,110) -- cycle ;
\draw    (490,390) -- (517,390) ;
\draw [shift={(520,390)}, rotate = 180] [fill={rgb, 255:red, 0; green, 0; blue, 0 }  ][line width=0.08]  [draw opacity=0] (8.93,-4.29) -- (0,0) -- (8.93,4.29) -- cycle    ;
\draw    (190,90) -- (190,63) ;
\draw [shift={(190,60)}, rotate = 90] [fill={rgb, 255:red, 0; green, 0; blue, 0 }  ][line width=0.08]  [draw opacity=0] (8.93,-4.29) -- (0,0) -- (8.93,4.29) -- cycle    ;
\draw  [dash pattern={on 4.5pt off 4.5pt}]  (190,290) -- (400,90) ;

\draw (352,233.4) node [anchor=north west][inner sep=0.75pt]    {$\overline{p}$};
\draw (322,263.4) node [anchor=north west][inner sep=0.75pt]    {$p_{L}$};
\draw (382,203.4) node [anchor=north west][inner sep=0.75pt]    {$p_{U}$};
\draw (400,66.6) node [anchor=south] [inner sep=0.75pt]    {$GFT=6h$};
\draw (162.4,259) node [anchor=north west][inner sep=0.75pt]  [rotate=-270]  {buyer};
\draw (321,392.4) node [anchor=north west][inner sep=0.75pt]    {seller};
\draw (522,393.4) node [anchor=north west][inner sep=0.75pt]    {$s$};
\draw (188,56.6) node [anchor=south east] [inner sep=0.75pt]    {$b$};

\draw (185,100) node [anchor=south east] [inner sep=0.75pt]    {$1$};

\draw (185,393.4) node [anchor=north east] [inner sep=0.75pt]    {$0$};

\draw (490,394.) node [anchor=north west] [inner sep=0.75pt]    {$1$};

\end{tikzpicture}}
\vspace{-0.2cm}\caption{The $x$-axis (resp., $y$-axis) represents the seller's (resp., buyer's) valuations.
The blue (resp., yellow) region denotes valuations for which $p_L$ (resp., $p_U$) is rejected. \Cref{lem:afterRA} states that if there are valuations both in the yellow and blue one, then there cannot be any with $\GFT \ge 6h=6L2^{-\ell}$.
Let the green and red regions be the image of $\area(N_{\ell,z})$ under $(f_s,f_b)$, then \Cref{lem:afterRA} forbids the red region, but allows the green one.
  }
  \label{fig:boat3}
\end{wrapfigure}
For an active node $N_{\ell,z}$, once the \reducearea procedure has terminated, the algorithm proceeds with the \blindguess routine (see \Cref{alg:guess}). In this phase, prices are drawn uniformly at random from the set $\mathcal{P}_{\ell,z} = [\bar p - L2^{-(\ell-1)}, \bar p + L2^{-(\ell-1)}]_{\epsilon}$, where $\bar p = p_{\ell-1,\pred(\ell,z)}$ is the price used to mark the parent node,  and $\epsilon$ will be specified later. 
In the \blindguess procedure, random prices are posted until a trade at some price $p$ is accepted, at which point $N_{\ell,z}$ is marked with $p_{\ell,z} = p$.

The set $\mathcal{P}_{\ell,z}$ is chosen so that, by the Lipschitz condition, it always contains a feasible price. Specifically, for every $x \in \area(N_{\ell,z})$ with $f_b(x) \ge f_s(x)$, there exists a price in the interval $[\bar p - L2^{-(\ell-1)}, \bar p + L2^{-(\ell-1)}]_\epsilon$ that will be accepted by the valuations $f_s(x)$ and $f_b(x)$.
Moreover, we observe that $|\mathcal{P}_\epsilon|=O(\epsilon^{-1}L2^{-\ell})$. This allows us to bound the expected number of rounds required to terminate the \blindguess procedure, and to upper bound the regret incurred by executing it. 

Formally, we denote by $T_{\ell,z}^\blindguess\subseteq [T]$ the set of rounds in which \Cref{alg:EPBT} executes $\blindguess(N_{\ell,z})$. Then, let $R_{\ell,z}^\blindguess\coloneqq \mathbb{E}\big[\mathcal{R}(T_{\ell,z}^\blindguess)\big]$ be the regret accumulated when \blindguess is executed for note $N_{\ell,z}$. We provide the following upper bound on this regret term.

\begin{restatable}[Regret due to \blindguess]{lemma}{lemmaregretBG}\label{lem:regretBG}
    Let $x\in \area(N_{\ell,z})$. For any $\epsilon>0$, consider posting prices uniformly at random from $\mathcal{P}_{\ell,z}=[\bar p-L2^{-(\ell-1)}, \bar p+L2^{-(\ell-1)}]_{\epsilon}$, where $\bar p=p_{\ell-1,\pred(\ell,z)}$ is the marking prices of the parent node. Then the regret incurred by the \blindguess procedure is at most
    \[
    R_{z,\ell}^{\blindguess}\le \epsilon \mathbb{E}[|T_{\ell,z}|]+24\frac{L^22^{-2\ell}}{\epsilon}.
    \]
\end{restatable}

The intuition about this result is simple: if the trades have small (\ie, less than $\epsilon$) $\GFT$ then the total regret is at most $\epsilon$ per round. On the other hand, if the $\GFT$ is large (\ie, greater then $\epsilon$) then an $\epsilon$-grid over the possible prices is enough to win a trade with probability $\Omega(\epsilon2^{\ell}/L)$, which happens on average after $O(L2^{-\ell}/\epsilon)$ turns. In each of these turns we are only loosing $O(L2^{-\ell})$ because they take place after the \reducearea procedure of that node.

\subsection{Regret analysis}

By using the results of the previous sections, and for an appropriate choice of $\epsilon$, we can now upper bound the overall regret incurred by the algorithm.

\begin{restatable}{theorem}{theoremregret}\label{th:UB}
    For any $T>2^d$ and any pair of $L$-Lipschitz valuation functions, \Cref{alg:EPBT} with $\epsilon=LT^{-\nicefrac{1}{d}}$ guarantees
 $R_T=O(L \log_2(T)T^{(d-1)/d})$.
\end{restatable}

The proof, provided in the appendix, begins by partitioning the rounds $[T]$ across nodes according to \Cref{eq:nodeDecomposition}. For each node, the rounds are then further subdivided based on the phase of the algorithm. Intuitively, for each round, we look at the node ``from which we post a price''. Then, such rounds are partitioned based on whether the algorithm invoked \reducearea, \blindguess, or if the node is a marked leaf, as specified in Line \ref{line:markedleaf} of \Cref{alg:EPBT}. Each resulting term can be suitably upper-bounded by applying the previous lemmas.

\section{Unknown Lipschitz Constant: a Multi-Scale Approach}\label{sec:noL}

The algorithm presented in \Cref{sec:algo}, heavily relies on the knowledge of the Lipschitz constant $L$. The exact knowledge of this constant is crucial in proving the correctness of the algorithm. More precisely, we exploited several Lipschitzness arguments to bound valuations. However, in many real-world scenarios it may be unrealistic to assume that $L$ is known. 
In this section, we present an algorithm that extends the one introduced in \Cref{sec:algo}, removing the need for a priori knowledge of $L$. This comes at the cost of only a slight deterioration in guarantees, yielding a regret of $\widetilde O(L^2 T^{(d-1)/d})$ instead of $\widetilde O(L T^{(d-1)/d})$.

We rely on geometric (\ie, multi-scale) extensions of the \reducearea and \blindguess subroutines, which we call \geometricreducearea and \geometricguess, respectively. In particular, for each node $N_{\ell,z}$, we define a sequence of geometrically increasing Lipschitz constants $\tilde L^j_{\ell,z}=O(2^j)$. In \geometricreducearea we simply perform the \reducearea routines for each scale $j$. Crucially, terminating this procedure bounds the $\GFT$ in all subsequent nodes by $O(L2^{-\ell})$ (similarly to \Cref{lem:afterRA}), but without requiring the knowledge of $L$.

The construction of \geometricreducearea is more involved, and The reason is clear if we think about the analysis of \blindguess. 
The analysis relied on a simple fact: for all trades with $\GFT$ greater than $\epsilon$, we have a probability of at least $\Omega(\epsilon2^{\ell}/L)$ of marking the node, and thus we expect to finish in $O(L2^{-\ell}/\epsilon)$ rounds.
However, it is also possible that the process never terminates, since there may be many trades happening with $\GFT$ less than $\epsilon$. 
Therefore, without knowledge of the constant $L$, the algorithm may spend many rounds working at an excessively small scale, potentially leading to the loss of trades with large $\GFT$. Moreover, it becomes impossible to distinguish whether we are losing these trades because they have small $\GFT$, or simply because the chosen scale is too small for the real Lipschitz constant.
The solution is to carefully increase the scale of our Lipschitz constant (more precisely, logarithmically), and to draw a price from the corresponding grid uniformly at random. This process is guaranteed to finish in expectation after $O(L2^{-\ell}/\epsilon)$ rounds against valuations with $\GFT$ at least $\epsilon$, and does not require knowing $L$ beforehand. 
The overall algorithm is detailed in \Cref{app:algo2} and provides the following regret guarantees, which are only a $O(\log_2(T)L)$ factor greater than the guarantees of \Cref{th:UB}.

\begin{restatable}{theorem}{theoremWL}
    Consider the nonparametric contextual bilateral trade problem, with one-bit feedback and strong budget balance with $L$-Lipschitz valuation functions. Then there is an algorithm that, without taking as input $L$, guarantees a regret of
    $
        R_T=O(\log_2(T)^2L^2 T^{(d-1)/d}).
    $
\end{restatable}

\section{Lower Bound}\label{sec:LB}

In this section, we proceed to prove a matching lower bound that proves the tightness of our algorithm. In particular, we will prove the following theorem.

\begin{theorem}\label{th:LB}
    Any randomized algorithm under strong-budget balance, $d$-dimensional context vectors, and $L$-Lipschitz value functions, suffers a regret of $\Omega(LT^{\frac{d-1}{d}})$, provided that $T>(4L)^d$. This also holds under full feedback.
\end{theorem}

From a technical perspective, the proof of \Cref{th:LB} relies on Yao's weak minimax principle \citep{yao1977probabilistic}, which lower bounds the regret with the expected performance of a deterministic algorithm against a distribution of instances. Formally:

\begin{lemma}[Yao's weak minimax principle]
For any deterministic algorithm generating $\{p_t\}_{t\in[T]}$, and any distribution $\mathcal{D}$, we have that
$R_T\ge \mathbb{E}_{(f_{s},f_b)\sim\mathcal{D}}\left[\sum_{t\in [T]}[f_b(x_t)-f_s(x_t)]^+-\GFT_t(p_t)\right]$.
\end{lemma}

To construct a hard distribution over instances, and consequently over valuation functions $f_s$ and $f_b$, we start by considering all $2^T$ possible assignments of a discrete set of $T$ context vectors to two disjoint valuations (\ie, such that no price results in both trades being accepted). We then extend these discrete assignments to Lipschitz-continuous functions over the entire context space.
The distribution is then taken as the uniform one over all $2^T$ possible valuation functions.
In particular, we build a distribution $\mathcal{D}$ over instances such that, for any deterministic algorithm, it holds 
\[
\textstyle
\mathbb{E}_{(f_{s},f_b)\sim\mathcal{D}}\Big[\sum_{t\in [T]}[f_b(x_t)-f_s(x_t)]^+-\GFT_t(p_t)\Big]\ge \widetilde\Omega(LT^{\frac{d-1}d}).
\]

We define the standard $\delta$-grid $\mathcal{G}$ on the $[0,1]^d$ dimensional hypercube, \ie, $\mathcal{G}=[0,\delta,\ldots, 1]^d$. The set $\mathcal{G}$ is the set of possible contexts and, by fixing $\delta={T^{-1/d}}$, we have that $|\mathcal{G}|=T$. Fix any ordering of contexts $x_1,\ldots,x_T$ over $\mathcal{G}$. We can define two valuations 
\(
(s_0,b_0) = \left(\frac{1}{2}-\gamma,\frac{1}{2}-\epsilon\right)
\quad\text{and}\quad
(s_1,b_1) = \left(\frac{1}{2}+\epsilon,\frac{1}{2}+\gamma\right),
\)
where $\gamma>\epsilon>0$ are two parameters to be specified later. Notice that there is no price $p$ such that $p$ wins both trades, in particular when the valuations are $(s_0,b_0)$ the only prices that get the trade accepted are in the set $[\frac{1}{2}-\gamma,\frac12-\epsilon]$, while under valuations $(s_1,b_1)$ prices that make the trade happen are in $[\frac12+\epsilon,\frac12+\gamma]$. For the same reason, we know that we can focus on deterministic algorithms that post prices only within one of these two intervals, since they yield $\GFT$ greater than that of any other price.

Next, for any $L>0$, we build a family of $L$-Lipschitz functions from $[0,1]^d$ to $\{(s_0,b_0),(s_1,b_1)\}$.
In order to build our set of instances, consider a binary vector $h\in\mathcal{H}_T$ of length $T$. Given a context $x_i$, the $i$-th bit of $h$ determines the valuations for such a context. In particular, $h$ selects $(s_0,b_0)$ if $h_i=0$, and $(s_1,b_1)$ if $h_i=1$.
For each of these binary vectors $h\in\mathcal{H}_T$, let $f^{h}=(f_s^h,f_b^h):\mathcal{G}\to \Reals^2$ be the functions satisfying the corresponding assignment, \ie,
$f^h_s(x_i)=s_{h_i}\quad\text{and}\quad f^h_b(x_i)=b_{h_i}$.
Note that for all $x\neq x'\in \mathcal{G}$ and any $h\in\mathcal{H}$, we have  $|f^h_s(x)-f^h_s(x')|\le \gamma-\epsilon$, and the same holds for $f_b^h$. By taking $\epsilon=\frac{\gamma}{2}$ and $\gamma=2LT^{-1/d}$, and by noting that $\inf_{x\neq x', x\in \mathcal{G}}\|x-x'\|_\infty\ge \delta=T^{-1/d}$ we have that $f^h_s$ and $f^h_b$ are $L$-Lipschitz in $\mathcal{G}$.\footnote{Note that we need $T>(4L)^d$ for the valuations to be well defined (\ie, in $[0,1]^2$).} Now, we need to extend the functions over the entire context domain $[0,1]^d$, which can be done in a black box way by the McShane extension theorem that we recall here
\begin{theorem}[McShane's extension \citep{mcshane1934extension}]\label{th:extension}
    Let $X\subset \Reals^d$ be a closed set, and $f$ be a $\Reals$-valued, $L$-Lipschitz function on $X$. Then there exists an extension $\tilde f:\Reals^d\to \Reals$ that agrees with $f$ on $X$ and is $L$-Lipschitz on $\Reals^d$.\footnote{In our proof we only care about the existence of such an extension, however, the proof of the theorem is constructive and defines $\tilde f$ as $\tilde f(x)=\inf_{x'\in X}(f(x)+L\|x-x'\|_\infty)$, \cite[][see \eg, Theorem 1.33]{weaver2018lipschitz}.}
\end{theorem}

We apply this on each $f^h$ and take $\mathcal{D}$ as the uniform distribution over $\mathcal{F}=\{\tilde f^{h}\}_{h\in \mathcal{H}^T}$. %
This distribution has the intuitive property that renders its history unhelpful in predicting future valuations.
\begin{restatable}{lemma}{lemmauniform}\label{lem:uniform}
    For any $t\in[T]$, the distribution of $(s_t,b_t)$ under $\mathcal{D}$, conditioned on the filtration up to time $t-1$, is uniform over $\{(s_0,b_0),(s_1,b_1)\}$.
\end{restatable}

We start by observing that the $\GFT$ that can be extracted from both $(s_0,b_0)$ and $(s_1,b_1)$ is $\gamma-\epsilon=\gamma/2>0$. Therefore, the benchmark will gain $\mathbb{E}_{(f_s,f_b)\sim\mathcal{D}}\big[\sum_{t\in[T]}[f_b(x)-f_s(x)]^+\big]=\frac\gamma2T$.
On the other hand, by \Cref{lem:uniform}, for any $t\in[T]$ the best expected $\GFT$ achievable by any deterministic algorithm is bounded from above by $\frac{\gamma-\epsilon}{2}=\frac\gamma4$.
Since we set $\gamma=2LT^{-1/d}$, we obtain that
\[
R_T\ge \frac{T}{2}(\gamma-\epsilon)=\frac{L}{2} T^{\frac{d-1}d},
\]
 concluding the proof of \Cref{th:LB}.

 \paragraph{Remark} This construction is reminiscent of that of \citet{lu2009showing}, who establish a rate of $\Omega(T^{(d+1)/(d+2)})$ (instantiated, as an example, in the finite MAB setting), which is higher than both our lower bound of $\Omega(T^{(d-1)/d})$ and our upper bound of $\widetilde O(T^{(d-1)/d})$, despite our setting appearing more challenging. However, this is not unexpected, since the proof of \citet{lu2009showing} leverages a significantly sparser grid in the context space ($O(T^{-1/(d+2)})$ instead of $O(T^{-1/d})$). They are allowed to do this since in their setting contexts can be repeated without incurring zero regret, as noise allows suboptimal actions to remain indistinguishable for longer. In contrast, in our setting, the reward function is identical whenever the same context is observed.

\section{Conclusions and Future Work}

There are many interesting possible directions for future works. Similar to \citet{slivkins2014contextual}, one can consider a general context space and derive rates that depend on its packing or covering dimension. Alternatively, following \citet{kuzborskij2020locally, cesa2017algorithmic}, one may exploit additional structure in the function class mapping contexts to valuations. From a bilateral trade perspective, other notions of budget balance are also relevant, such as {weak} and {global} budget balance \citep{cesa2024bilateral, bernasconi2023no}. Under these notions, the learner may post different prices to the buyer and the seller, subject only to the constraint that the market is not subsidized in each round and across all rounds, respectively. Moreover, our tight rates are attainable in part because our model does not allow for noise, unlike some general nonparametric bandit models \citep{lu2009showing, slivkins2014contextual}. The question of what rates can be achieved under the noisy valuations remains open. 
Finally, \citet{Cosson2026better} also studies profit maximization. It would be interesting to study the profit maximization setting  with general Lipschitz contexts.

\section*{Acknowledgment}
The work of MB, and AC was partially funded by the European Union. Views and opinions expressed
are however those of the author(s) only and do not necessarily reflect those of the European Union or the
European Research Council Executive Agency. Neither the European Union nor the granting authority can
be held responsible for them.

This work is supported by an ERC grant (Project 101165466 — PLA-STEER).

\clearpage

\printbibliography

\clearpage
\appendix
\section{Omitted proofs from \Cref{sec:UB} (Upper Bound)}
\lemmaafterRA*

\begin{proof}
    Let $p_L=\Pi_{[0,1]}(\bar p-L2^{-(\ell-1)})$ and $p_U=\Pi_{[0,1]}(\bar p+L2^{-(\ell-1)})$, where $\bar p= p_{\ell-1,\pred(\ell,z)}$ is the price with which we marked the parent node.

    If \reducearea terminated for node $N_{\ell,z}$, then Line \ref{line:pLpU} of \Cref{alg:reducearea} has been executed exactly twice. We denote by $t_L\in [T]$ the first time in which Line \ref{line:pLpU} is executed (\ie, the round in which we observe a trade rejected for $p_L$), and analogously we denote by  $t_U\in [T]$ the second time in which Line \ref{line:pLpU} is executed (\ie, the round in which we observe a trade rejected at price $p_U$).
    Let $x_L=x_{t_L}$ and $x_U=x_{t_U}$.
    
    At the end of \reducearea we have two pairs $(p_L,x_L)$ and $(p_U,x_U)$ such that 
    \[
    f_s(x_L)>p_L \text{ or } f_b(x_L)<p_L,
    \]
    and 
    \[
    f_s(x_U)>p_U \text{ or } f_b(x_U)<p_U.
    \]
    Moreover, by the Lipschitz condition and the fact that the parent node $N_{\ell-1,\pred(\ell,z)}$ is marked (see \Cref{def:marked}), it follows that for every $x\in \area(N_{\ell,z})\subseteq \area(N_{\ell-1,\pred(\ell,z)})$ we have 
    \[
    f_s(x)\le \min\{1,\bar p+L2^{-(\ell-1)}\}=p_U\quad \text{ and }\quad f_b(x)\ge \max\{0,\bar p-L2^{-(\ell-1)}\}2=p_L.
    \]
    
    Then, by rewriting the above for $x_U\in \area(N_{\ell,z})$, we get that the buyer refused the trade at price $p_U$, since the seller would have accepted it. Hence, $f_b(x_U)<p_U$. Similarly, we know that it was the seller who rejected the trade at $x_L$, implying $f_s(x_L)>p_L$.
    Now consider any $x\in \area(N_{\ell,z})$. By Lipschitzness we have that
    $$|f_s(x)-f_s(x_L)|\le L2^{-\ell}\quad\text{ and }\quad |f_b(x)-f_b(x_U)|\le L2^{-\ell},$$ from which we can conclude that
    \[
    f_b(x)-f_s(x)\le f_b(x_U)-f_s(x_L)+2L2^{-\ell}\le p_U-p_L+2L2^{-\ell}=6L2^{-\ell}.
    \]
    This concludes the proof.
\end{proof}

\lemmaregretRA*

\begin{proof}
    First, observe that in rounds belonging to $T_{\ell,z}^\reducearea$ we accumulate regret only twice, once for each rejection of prices $p_L$ and $p_U$, since accepted prices yield zero regret.
    Moreover, by construction we have that $N_{\ell,z}$ is active, so we know that the parent node $N_{\ell-1, \pred(\ell,z)}$ has been marked and that \reducearea has been completed for $N_{\ell-1, \pred(\ell,z)}$. For any $x\in \area(N_{\ell,z})\subseteq \area(N_{\ell-1,\pred(\ell,z)})$ we can apply \Cref{lem:afterRA}, and we get that $f_b(x)-f_s(x)\le 6L2^{-(\ell-1)}=12L2^{-\ell}$.
    Therefore, we obtain $R_{\ell,z}^\reducearea\le 24 L2^{-\ell}$, which concludes the proof.
\end{proof}

\lemmaregretBG*
\begin{proof}
    For each $x\in \area(N_{\ell,z})$ we have that
    \[
    f_s(x)\le f_s(\bar x)+L2^{-(\ell-1)}\le \bar p+L2^{-(\ell-1)},
    \]
    where $\bar x\in \area(N_{\ell-1,\pred(\ell,z)})$ is the context in the parent node's area under which $\bar p$ was accepted. 
    Similarly, we have that
    \[
    f_b(x)\ge f_b(\bar x)-L2^{-(\ell-1)}\ge\bar p -L2^{-(\ell-1)}. 
    \]
    Thus, if $f_b(x)\ge f_s(x)$ then the interval $[f_s(x),f_b(x)]$ is contained in the interval $[\bar p -L2^{-(\ell-1)}, \bar p +L2^{-(\ell-1)}]$.

    For any $t\in T_{\ell,z}^\blindguess$ we are going to consider two cases: (i) $f_b(x_t)-f_s(x_t)\le \epsilon$, and (ii) $f_b(x_t)-f_s(x_t)> \epsilon$. We let 
    \[
T_{\ell,z}^{\blindguess,1}\coloneqq \left\{t\in T_{\ell,z}^\blindguess: f_b(x_t)-f_s(x_t)\le \epsilon\right\},
    \]
    and $T_{\ell,z}^{\blindguess,2}\coloneqq T_{\ell,z}^{\blindguess}\setminus T_{\ell,z}^{\blindguess,1}$.
    Then, by linearity of expectation, the regret $R_{z,\ell}^{\blindguess}$ can be decomposed as
    \begin{align*}
    R_{z,\ell}^{\blindguess}%
    =\mathbb{E}\left[\mathcal{R}\left(T_{\ell,z}^{\blindguess,1}\right)\right]+\mathbb{E}\left[\mathcal{R}\left( T_{\ell,z}^{\blindguess,2}\right)\right].
    \end{align*}

First, we observe that if $t \in T_{\ell,z}^{\blindguess,1}$, then the first term in the decomposition can be bounded above by $\epsilon \, \mathbb{E}[|T_{\ell,z}|]$. Indeed, in this case we have $[f_b(x_t) - f_s(x_t)]^+ \le \epsilon$,  $\GFT_t(\cdot \mid f_s(x_t), f_b(x_t)) \ge 0$, and moreover $T_{\ell,z}^{\blindguess,1} \subseteq T_{\ell,z}$.

Similarly, the second term can be upper bounded by $\mathbb{E}[|T_{\ell,z}^{\blindguess,2}|]\cdot G$, where $G$ is the maximum $\GFT$ attainable. Since the \reducearea procedure must have already terminated for $N_{\ell,z}$ before \blindguess is executed, we can apply \Cref{lem:afterRA}, which yields $G \le 6L2^{-\ell}$. Therefore,
    \[
    R_{z,\ell}^{\blindguess}\le \epsilon \mathbb{E}[|T_{\ell,z}|]+ 6L2^{-\ell}\cdot \mathbb{E}[|T_{\ell,z}^{\blindguess,2}|].
    \]

    In order to conclude the proof we have to upper bound $\mathbb{E}[|T_{\ell,z}^{\blindguess,2}|]$.
    Notice that for every $t \in T_{\ell,z}^{\blindguess,2}$ we have $f_b(x_t) - f_s(x_t) > \epsilon$. Hence, there exists at least one price $p$ in the discretized grid $\mathcal{P}_{\ell,z}$ such that $f_s(x_t) \le p \le f_b(x_t)$. This implies that the probability of terminating \blindguess at each round $t \in T_{\ell,z}^\blindguess$ is at least $|\mathcal{P}_{\ell,z}|^{-1}$, and therefore
\[
\mathbb{E}[|T_{\ell,z}^{\blindguess,2}|] \le |\mathcal{P}_{\ell,z}| \le \tfrac{4L2^{-\ell}}{\epsilon}.
\]
This concludes the proof.
\end{proof}

\theoremregret*
\begin{proof}

    For a node $N_{\ell,z}$, let $T^\leaf_{\ell,z}\subseteq [T]$ be the set of rounds in which Line \ref{line:markedleaf} of \Cref{alg:EPBT} is executed for that node (\ie, the number of times a leaf node is played after being marked). Moreover, let 
    \[
    R^{\leaf}_{\ell,z}\coloneqq \mathbb{E}\left[\mathcal{R}(T^\leaf_{\ell,z})\right].
    \]

First, following \Cref{eq:nodeDecomposition}, we decompose the overall regret into the contributions incurred at each level and at the nodes within that level:
\begin{align*}
    R_T&=\sum_{\ell=0}^{H}\sum_{z\in Z_\ell} R_{\ell,z}=\sum_{\ell=0}^{H}\sum_{z\in Z_\ell} (R^\reducearea_{\ell,z}+R_{\ell,z}^\blindguess+R_{\ell,z}^\leaf)
\end{align*}
where the second equality come from the fact that $T_{\ell,z}^\reducearea,T_{\ell,z}^\blindguess,T_{\ell,z}^\leaf$ form a partition of $T_{\ell,z}$ for each level $\ell\in\{0,1,\ldots,H\}$ and node $z\in Z_\ell$.

\paragraph{\reducearea} Now we are going to analyze the first term.
By \Cref{lem:regretRA} we have that $R_{\ell,z}^\reducearea\le 24 \cdot L2^{-\ell}$. Moreover, we recall that $H=\lfloor \log_{2^d}T\rfloor$ and  $|Z_\ell|=2^{d\ell}$. Thus
\begin{align*}
\sum_{\ell=0}^{H}\sum_{z\in Z_\ell} R^\reducearea_{\ell,z}&\le 24\cdot L \sum_{\ell=0}^{\lfloor\log_{2^d}(T)\rfloor} 2^{d\ell-\ell}\\
&\le \frac{24 L \log T}{d} 2^{dH}\cdot 2^{-H}\\
&\le \frac{48 L \log T}{d} T^{1-\frac1d},
\end{align*}
where in the last inequality we used that that $\log_{2^d}(T)= \frac{1}{d}\log(T)$ and that $\ell\mapsto2^{d\ell-\ell}$ is monotone so that %
$\sum_{\ell=0}^{H}2^{d\ell-\ell}\le H 2^{d{H}-H}$.

\paragraph{\blindguess} Similarly, for the second term, we are going to proceed as follows:

\begin{align*}
    \sum_{\ell=0}^{H}\sum_{z\in Z_\ell} R_{\ell,z}^\blindguess &\le \sum_{\ell=0}^{H}\sum_{z\in Z_\ell} \left(\epsilon \mathbb{E}[|T_{\ell,z}|]+24\frac{L^2 2^{-2\ell}}{\epsilon}\right)\tag{\Cref{lem:regretBG}}\\
    &=\epsilon T+\frac{24L^2}{\epsilon}\sum_{\ell=0}^{H} 2^{d\ell-2\ell}\tag{$|Z_\ell|=2^{d\ell}$}\\
    &\le \epsilon T+\frac{24L^2\log T}{\epsilon d} 2^{d\log_{2^d}(T)}\cdot 2^{-2\log_{2^d}(T)}\\
    &\le \epsilon T+\frac{24L^2\log T}{\epsilon d} T^{1-\frac2d}.
\end{align*}

Finally, by choosing $\epsilon=L T^{-\frac{1}{d}}$
we obtain 
\[
\sum_{\ell=0}^{H}\sum_{z\in Z_\ell} R_{\ell,z}^\blindguess\le L T^{1-\frac1d} + \frac{24 L}{ d} \log(T) T^{1-\frac{1}{d}} \le \frac{d+24}{d}L\log(T) T^{1-\frac{1}{d}}.
\]

\paragraph{\leaf} To analyze the last term, we recall that $T_{\ell,z}^\leaf$ is non-empty only if $\ell=H$ (\ie, it is a leaf node) and the node was already marked. Thus we can use \Cref{lem:afterRA} which shows that 
\[
R_{\ell,z}^\leaf\le 6L2^{-\log_{2^d}(T)}\mathbb{E}\left[T_{\ell,z}^\leaf\right]\le 6LT^{-1/d} \mathbb{E}\left[T_{\ell,z}\right].
\] 
Therefore, by restricting our attention to the $H$-th layer we get
\begin{align*}
\sum_{\ell=0}^{H}\sum_{z\in Z_\ell} R_{\ell,z}^\leaf & = \sum_{z\in Z_H} R_{H,z}^\leaf\\
&\le \sum_{z\in Z_H} 6LT^{-1/d} \mathbb{E}\left[T_{H,z}\right]\\
&=6LT^{\frac{d-1}{d}}.
\end{align*}

\paragraph{Combining all parts} Thus, the total regret can be bounded by 
\[
R_T\le \frac{72+7d}{d} L T^{1-\frac1d}\log(T).
\]
concluding the proof.
\end{proof}

\section{Omitted proofs from \Cref{sec:LB} (Lower Bound)}
\lemmauniform*
\begin{proof}
    Since we have full feedback, we know the sequence of valuations $\{(s_\tau,b_\tau)\}_{\tau\in[t-1]}$. The set of functions $\tilde f^h$ that agree on these valuations for the sequence of observed contexts is $2^{t-1}$, we call this set $\tilde {\mathcal{F}}_t$. Out of these functions, half will have $h_t=0$ and half will have $h_t=1$, and thus the conditioned probability on $(s_t,b_t)$ is uniform over $\{(s_0,b_0),(s_1,b_1)\}$.
\end{proof}

\section{Description of the Algorithm from \Cref{sec:noL} and Omitted Proofs}\label{app:algo2}

The algorithm's main loop is as the vanilla one introduced in \Cref{sec:UB}. In particular, we traverse the tree until we reach an active node (Line \ref{line:traverse2} of \Cref{alg:EPBT2}). Then we perform either the \geometricreducearea of \Cref{alg:reducearea2} or \geometricguess of \Cref{alg:guess2}. In particular, for a specific node, we continue to perform a step of the main loop of \Cref{alg:reducearea2} until it meets its termination condition, or one step of \Cref{alg:guess2}, until it terminates. \Cref{alg:reducearea2} simply tries sequentially all scales of $\tilde L^j$ as in the plain version. 
\Cref{alg:guess2}, on the other hand, at iteration $\tau$, uses the same largest scale permissible at that iteration Line \ref{line:largestscale} of \Cref{alg:guess2}, and then draws uniform prices from the associated grid (Line \ref{line:random}).

\begin{algorithm}[t]
\caption{Without knowing the Lipschitz costant}\label{alg:EPBT2}
\begin{algorithmic}[1]
\Require Horizon $T$, context dimension $d$.
\For{ $t\in [T]$}
\State Observe context $x_t$
\State Initialize $\ell\gets 0,z\gets \bm{0}_d$ %
\While{$p_{\ell,z}\neq \varnothing$ and $N_{\ell,z}$ not leaf}\hspace{-0.8cm}\hfill\Comment{Traverse tree until an active or leaf node is reached}\label{line:traverse2}
\State Find $z'$ such that $x_t\in \area(N_{\ell+1,z'})$, where $z'\in\child(\ell,z)$
\State $\ell\gets\ell+1$, $z\gets z'$
\EndWhile
\If{$\geometricreducearea(N_{\ell,z})$ has not terminated}
\State Execute one iteration of the main loop of $\geometricreducearea(N_{\ell,z})$%
\ElsIf{$\geometricguess(N_{\ell,z})$ has not terminated}
\State Execute one iteration of the main loop of $\geometricguess(N_{\ell,z})$%
\If{$\geometricguess(N_{\ell,z})$ terminates returning price $p$}
\State Set node $N_{\ell,z}$ as marked by setting $p_{\ell,z}\gets p$%
\EndIf
\EndIf
\State \algorithmicelse~Play $p_{\ell,z}$ \hfill\Comment{If node is a marked leaf play corresponding price}%
\EndFor
\end{algorithmic}
\end{algorithm}

\begin{algorithm}[t]
\caption{$\geometricreducearea(N_{\ell,z})$}\label{alg:reducearea2}
\begin{algorithmic}[1]
\State Set $\bar p\gets p_{\ell-1,\pred(\ell,z)}$ (\ie, the price associated to the parent node) \hspace{-0.8cm}\hfill\Comment{Initialization}
\State Set $\tilde L^j=L_02^j$ for $j\in\{0,\ldots,\bar j\}$ with $\bar j=\lceil\ell+\log_2(2L_0)\rceil$ and $L_0=1/2$
\For{$j\in\{0,\ldots,\bar j\}$} \hspace{-0.5cm}\hfill\Comment{Main Loop}
\State Set $p^j_L\gets\Pi_{[0,1]}\left(\bar p-\tilde L^j2^{-(\ell-1)}\right)$, $p^j_U\gets\Pi_{[0,1]}\left(\bar p+\tilde L^j2^{-(\ell-1)}\right)$
\State Initialize $p\gets p^j_L$ 
\Repeat
\State Post price $p$ and observe $\mathbb{I}(s_t\le p \le b_t)$
\If{$\mathbb{I}(s_t\le p \le b_t)=0$}
\State \algorithmicif~$p = p^j_L$ \algorithmicthen~set $p \gets p^j_U$; \algorithmicelse~terminate $p\gets\varnothing$%
\EndIf
\Until{$p=\varnothing$}
\EndFor
\end{algorithmic}
\end{algorithm}

\begin{algorithm}[t]
\caption{$\geometricguess(N_{\ell,z})$}\label{alg:guess2}
\begin{algorithmic}[1]
\State Let $\bar p=p_{\ell-1,\pred(\ell,z)}$ be the price associated to the parent node $N_{\ell-1,\pred(\ell,z)}$ \hspace{-1.1cm}\hfill\Comment{Initialization}
\State Set $\tilde L^j=L_02^j$ for $j\in\{0,\ldots,\bar j\}$ with $\bar j=\lceil\ell+\log_2(2L_0)\rceil$ and $L_0=1/2$
\State Set $\mathcal{P}_{\ell,z}^j=[\bar p-\tilde L^j2^{-(\ell-1)}, \bar p+\tilde L^j2^{-(\ell-1)}]_{\epsilon}$ with $\epsilon=T^{-\nicefrac{1}{d}}$
\State Set $\tau=0$
\Repeat\hspace{-0.5cm}\hfill\Comment{Main Loop}
\State $j^{(\tau)}:=\arg\max\big\{j\in\{0,\ldots,\bar j\}:|\mathcal{P}_{\ell,z}^j|\le \tau\big\}$\label{line:largestscale}
\State Post price $p\sim\text{Unif}(\mathcal{P}^{j^{(\tau)}}_{\ell,z})$ and observe $\mathbb{I}(s_t\le p \le b_t)$\label{line:random}
\State $\tau\gets\tau+1$
\Until{$\mathbb{I}(s_t\le p \le b_t)=1$}
\State Return $p$
\end{algorithmic}
\end{algorithm}

\begin{restatable}{lemma}{lemmatmpuno}\label{lem:tmp1}
    For every $\epsilon>0$ and for every node $N_{\ell,z}$, there exists an algorithm that does not require knowledge of $L$ such that \[R_{\ell,z}=O\left(\ell L2^{-\ell}+\epsilon\mathbb{E}[T_{\ell,z}]+\frac{L^22^{-2\ell}}{\epsilon}\right).\]
\end{restatable}

\begin{proof}
    Define $\bar j=\lceil\ell+\log_2(2L_0)\rceil$, and  for each $j\in\{0,\ldots,\bar j\}$ define $\tilde L^j=L_02^{j}$. Note that $2^{-(\ell-1)}\tilde L^{\bar j}>1$.
    For each $j$ we now start \geometricreducearea on prices $p_L^j=\bar p-2^{-(\ell-1)}\tilde L^j$ and $p_U^j=\bar p+2^{-(\ell-1)}\tilde L^j$, when we get a rejection both for $p_L^j$ and $p_U^j$ we increase $j$. %

\begin{claim}\label{claim:afterGRA}
    Once the \geometricreducearea procedure is concluded, then we know that for all $x\in \area(N_{\ell,z})$ we have $f_b(x)-f_s(x)\le 10 L 2^{-\ell}$.
\end{claim}
\begin{proof}[Proof of \Cref{claim:afterGRA}]
    If the \geometricreducearea has concluded we know that we collected tuples $(x^j_L,p_L^j)_{j=0,\ldots,\bar j}$ and $(x^j_U,p_U^j)_{j=0,\ldots,\bar j}$ that were refused. 

    Now take any $x\in\area(N_{\ell,z})$, since price $(\bar x,\bar p)$ was accepted at $N_{\ell-1,\pred(\ell,z)}$ we have that
    \[
    f_s(x)\le f_s(\bar x)+L 2^{-(\ell-1)}\le \bar p+L 2^{-(\ell-1)}=\bar p+2L 2^{-\ell},
    \]
    and 
    \[
    f_b(x)\ge f_b(\bar x)-L 2^{-(\ell-1)}\ge \bar p-L 2^{-(\ell-1)}=\bar p-2L 2^{-\ell}.
    \]
    We take $j'$ such that $\tilde L_{\ell-1,z}^{j'}\ge L$. Then, by plugging in the above $x_L^{j'}$ and $x_U^{j'}$, similarly to the proof of \Cref{lem:afterRA}, we know that
    \[
    f_b(x)-f_s(x)\le f_b(x_U^{j'})-f_s(x_L^{j'})+2L2^{-\ell}\le p_U^{j'}-p_L^{j'}+2L2^{-\ell}\le 4 \tilde L_{\ell-1,z}2^{-\ell}+2L2^{-\ell},
    \]
    where the first inequality is by Lipschitzness of $f_s,f_b$. Moreover, note that $2L\ge \tilde L^{j'}$ and thus $f_b(x)-f_s(x)\le 10L2^{-\ell}$ for all $z\in\area(N_{\ell,z})$. This concludes the proof of \Cref{claim:afterGRA}.
\end{proof}
By performing \geometricreducearea at node $N_{\ell,z}$, we only suffer a regret of at most $2\bar j\le 2(\ell+\log_2(2L_0)+1)\le 4\ell$ if $\ell\ge 1$ and $L_0=1/2$ (the number of rejected trades), which multiplies the maximum regret suffered, which by \Cref{claim:afterGRA}, is at most $20L2^{-\ell}$, and thus a total regret of $80\ell L2^{-\ell}$.

After this phase, we start using \geometricguess in the following multi-scale fashion. Let $\tau$ be the ``internal time'' of the routine. At each execution $\tau$ of \geometricguess, we consider the largest integer $j^{(\tau)}$ so that $|\mathcal{P}_{\ell,z}^{j^{(\tau)}}|\le \tau$, where we defined $\mathcal{P}_{\ell,z}^j=[\bar p-\tilde L^j2^{-(\ell-1)},\bar p+\tilde L^j2^{-(\ell-1)}]_\epsilon$ and thus $|\mathcal{P}_{\ell,z}^{j}|=4 L_02^j2^{-\ell}/\epsilon$. Now, at each $\tau$ we draw uniformly at random a price from the corresponding uniform grid $\mathcal{P}_{\ell,z}^{j^{(\tau)}}$. Define $\tau^\star$ as the smallest integer so that $L_02^{j^{(\tau)}}\ge L$, and note that $\tau^\star\le \frac{8L2^{-\ell}}{\epsilon}$.

After $\tau^\star$ turns of \geometricguess, we will have that, for all valuations with $\GFT$ at least $\epsilon$, there is at least one price in $\mathcal{P}_{\ell,z}^{j^{(\tau)}}$ that wins it. We then consider blocks $\mathcal{B}_i$ defined as the intervals $\mathcal{B}_i=[2^i\tau^\star, 2^{i+1}\tau^\star)$.
The probability of winning a trade with $\GFT$ larger than $\epsilon$ for $\tau\in \mathcal{B}_i$ is $(2^i\tau)^{-1}$, since for $\tau\in \mathcal{B}_i$ the size of $\mathcal{P}_{\ell,z}^{j^{(\tau)}}$ remains unchanged for all $\tau\in\mathcal{B}_i$ and smaller than $2^i\tau^\star$. We can define an appropriate Markov chain to compute the expected number of times until we stop.

\begin{claim}\label{claim:markov}
    Consider a Markov chain with states $\{0\}\cup\{N,N+1,\ldots\}$ and with transition probabilities $p_s=\mathbb{P}(s+1|s)=1/(2^kN)$ and $\mathbb{P}(0|s)=1-p_i$ for all $k\in\mathbb{N}$ and  $s\in[2^kN,\ldots, 2^{k+1}N-1]$. Moreover, let $\mathbb{P}(0|0)=1$. Then, the expected hitting time of the state $0$ starting from state $N$ is at most $4N$.
\end{claim}

\begin{proof}[Proof of \Cref{claim:markov}]
    The probability of reaching $s=2^{k+1}N$ from $s=2^k$ is $(1-1/(2^iN))^{2^{i}N}\le e^{-1}$ and the probability of reaching $s=2^kN$ from $s=N$ is thus at most $e^{-k}$. The expected number of transitions before hitting $0$ is the sum of the expected number of transitions in block $\mathcal{B}_i$, conditioned on reaching $2^iN$. We call $X_i$ the number of transitions inside of block $i$ and $T$ the hitting time of $s=0$. Thus:
    \begin{align*}
    \mathbb{E}[T]=\sum_{i=0}^\infty\mathbb{E}[X_i]&=\sum_{i=0}^\infty\mathbb{E}[X_i|\text{reach } 2^iN]\mathbb{P}(\text{reach } 2^iN)\tag{$\mathbb{E}\left[X_i\Big\vert\overline{\text{reach } 2^iN}\right]=0$}\\
    &\le \sum_{i=0}^\infty 2^iN e^{-i}\tag{$\mathbb{P}(\text{reach } 2^iN)\le e^{-i}$ and $X_i\le 2^iN$}\\
    &= N\sum_{i=0}^\infty\left(\frac2e\right)^i
    \end{align*}
    which is equal to $N\frac{e}{e-2}\le 4N$ concluding the proof of \Cref{claim:markov}.
\end{proof}

We can then use \Cref{claim:markov} to conclude that the expected number of turns of \geometricguess before observing a success is at most $4\tau^\star$, and each turn we experience a regret of only $20L2^{-\ell}$ thanks to \Cref{claim:afterGRA}.

Thus, for each $\tau\le \tau^\star$ we only have a regret of $10L2^{-\ell}$ per turn, and after (for $\tau\ge \tau^\star$), we have a regret of at most 
\begin{align*}
    \epsilon\mathbb{E}[T_{\ell,z}]+80\tau^\star L2^{-\ell}\le \epsilon\mathbb{E}[T_{\ell,z}]+ \frac{640L^22^{-2\ell}}{\epsilon},
\end{align*}
as in the proof of \Cref{lem:regretBG}.

The total regret for marking a node is
\begin{align}
R_{\ell,z}&\le \underbrace{80\ell L2^{-\ell}}_{\geometricreducearea}+\underbrace{10L2^{-\ell}}_{\tau\le\tau^\star}+\underbrace{\epsilon\mathbb{E}[T_{\ell,z}]+ \frac{640L^22^{-2\ell}}{\epsilon}}_{\tau>\tau^\star}\\
&=O\left(\ell L2^{-\ell}+\epsilon\mathbb{E}[T_{\ell,z}]+\frac{L^22^{-2\ell}}{\epsilon}\right).
\end{align}
\end{proof}

\begin{restatable}{lemma}{lemmatmpdue}\label{lem:tmp2}
    For every marked leaf node $N_{\ell,z}$ the regret $R_{\ell,z}$ is at most $O(L2^{-\ell}\mathbb{E}[T_{\ell,z}])$.
\end{restatable}
\begin{proof}
    Since we reached a leaf node, we must have concluded \geometricreducearea in the parent node and thus $R^\leaf_{\ell,z}\le 20L2^{-\ell}\mathbb{E}[T_{\ell,z}]$ by \Cref{claim:afterGRA}.
\end{proof}

\theoremWL*
\begin{proof}
    The regret of every non-marked node is upper bounded by \Cref{lem:tmp1}, while the leaves once marked have a regret upper bounded by \Cref{lem:tmp2}.
    Thus
    \begin{align*}    
    R_T&\le \sum_{\ell=0}^H\sum_{z\in Z_\ell} O\left(\ell L2^{-\ell}+\epsilon\mathbb{E}[T_{\ell,z}]+\frac{L^22^{-2\ell}}{\epsilon}\right)+\sum_{z\in Z_H}O\left(L2^{-H}\mathbb{E}[T_{z,H}]\right)\\
    &\le L\sum_{\ell=0}^HO(\ell2^{-\ell}2^{d\ell})+O(\epsilon T)+\frac{L^2}{\epsilon}\sum_{\ell=0}^H O(2^{d\ell}2^{-2\ell})+O(LT2^{-H})\\
    &\le O(LH^2T^{(d-1)/d})+O(\epsilon T)+O\left(\frac{L^2}{\epsilon}HT^{(d-2)/d}\right)+O(LT^{(d-1)/d}).
    \end{align*}
    then, by choosing $\epsilon=T^{-1/d}$, and noting that $H\le \log_2(T)/d$ we obtain that
    \[
    R_T\le O(\log_2(T)^2L^2 T^{(d-1)/d}),
    \]
    concluding the proof.
\end{proof}

\section{Experimental Evaluation}

For completeness, we provide an empirical validation on synthetic data to confirm our theoretical findings. 
\begin{wrapfigure}{l}{0.36\textwidth}
  \centering
  \includegraphics[width=\linewidth]{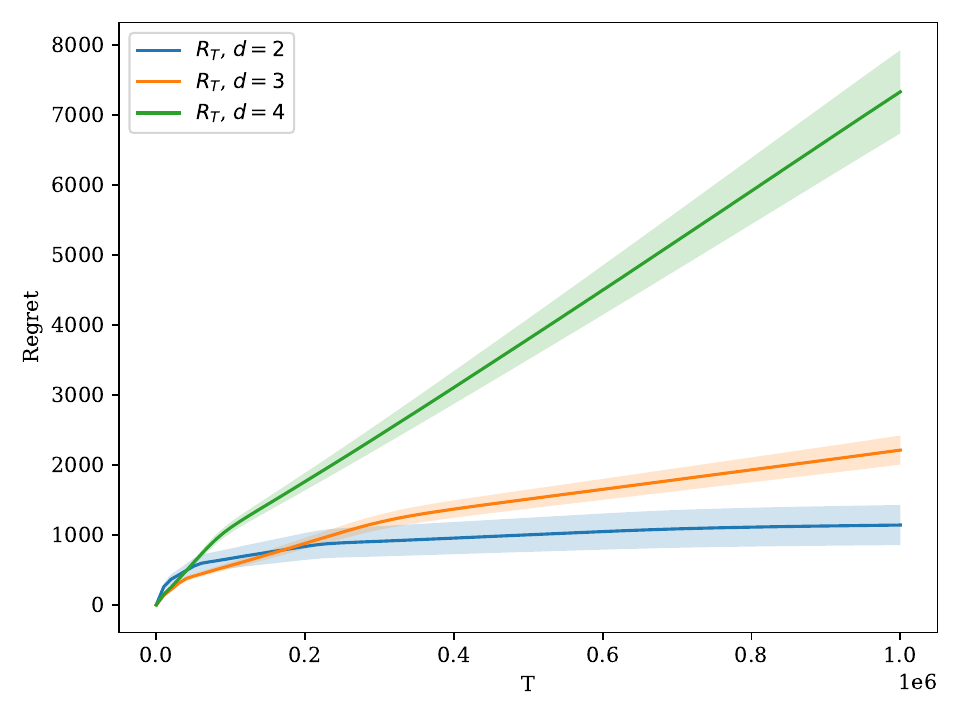}
  \caption{Average regret (together with 95\% confidence intervals) in the \emph{quadratic setting}.}\label{fig:exp}
\end{wrapfigure}
We consider a setting in which valuations are quadratic functions of the context. We run our algorithm on 30 independent instances built as follows. 
For each run we start by uniformly sampling two matrices $A,B\in [-1,1]^{d\times d}$. Then, for each $t$ we sample uniformly at random in $[0,1]^d$ a $d$-dimensional context $x_t$. We compute valuations as $x_t^\top A x_t$ and $x_t^\top B x_t$ up to normalization to have final valuations belong to the unit square. In particular, the valuations are computed as
\[
f_s(x_t)=\frac{x_t^\top A x_t + d^2}{d^2} \qquad f_b(x_t)=\frac{x_t^\top B x_t + d^2}{d^2}.
\]
We test our algorithms on $d=\{2,3,4\}$, $T= 10^6$, and run $30$ repetitions for each choice of $d$. The results are reported in \Cref{fig:exp}, and essentially confirm our theoretical findings. As expected, the learning problem becomes harder as the context dimension $d$ increases, and the behavior of average regrets is consistent with the rates obtained in \Cref{th:UB}.

\begin{figure}[!htbp]
  \centering
  \includegraphics[width=0.32\textwidth]{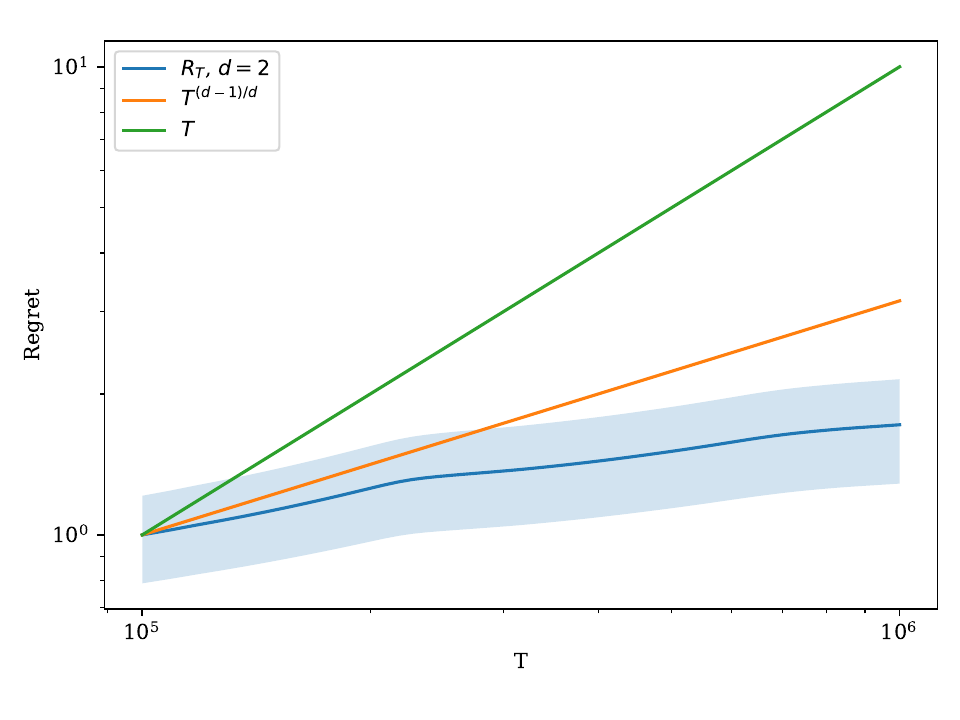}%
  \hfill
  \includegraphics[width=0.32\textwidth]{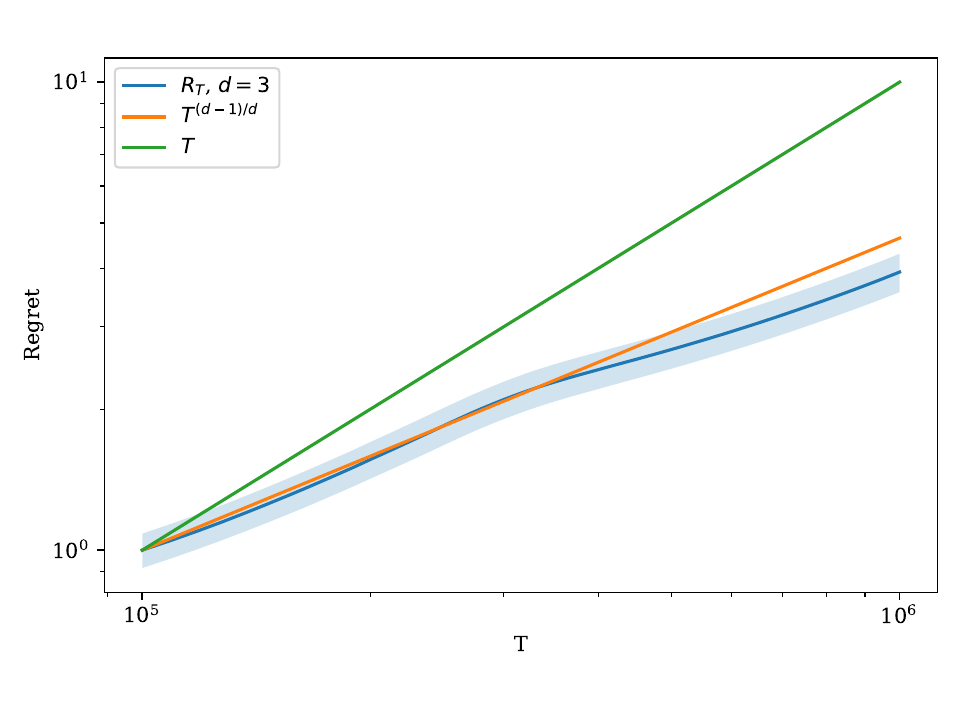}%
  \hfill
  \includegraphics[width=0.32\textwidth]{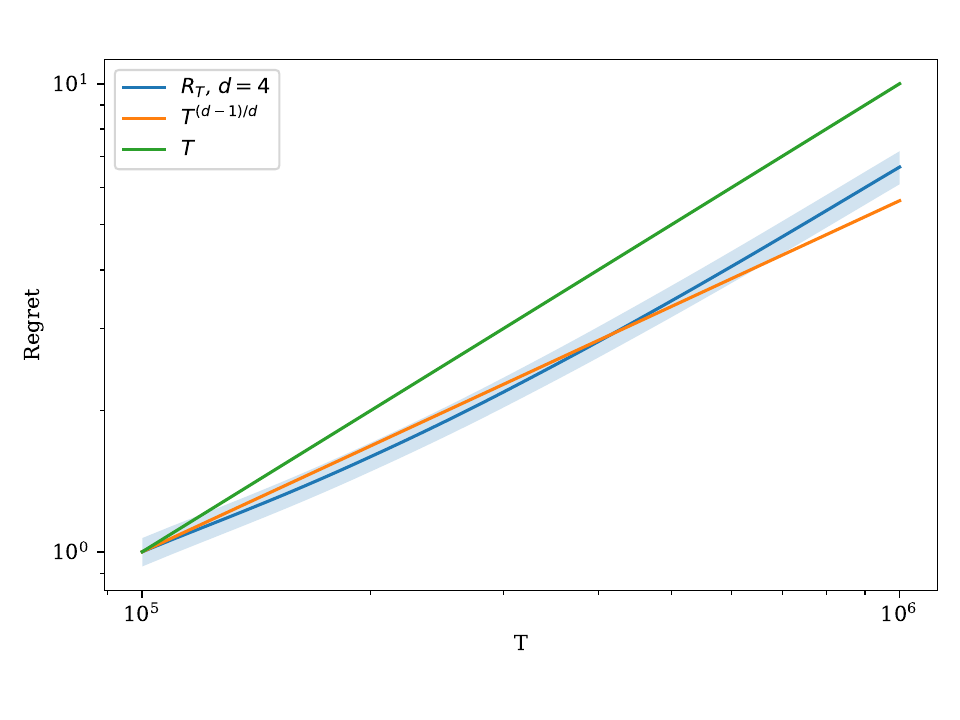}
  \caption{Comparison between average regret and theoretical bounds.}
  \label{fig:rates}
\end{figure}

In \Cref{fig:rates} we report the logarithmic plot of the regret. In particular, for the last $10^4$ steps, we plot $\log(t)$ vs $\log(R_t)$ for $d\in\{2,3,4\}$, along with the slopes associated with the theoretical bounds (i.e., $t^{(d-1)/d}$) and the linear slope $t$ (we normalize the lines so that they all start from $1$ so that we can compare more easily their slopes). We can see that the regret slope matches the theoretical guarantees. In particular, the regret grows at most as fast as $O(t^{(d-1)/d})$ for all $d$ (and for $d=2$ it appears to be growing even at a better rate).

\end{document}